\documentclass[a4paper,UKenglish]{lipics-v2016}
\newcommand{\ceil}[1]{\left \lceil #1 \right \rceil}
 
\usepackage{microtype}

\title{Optimal In-place Algorithms for Basic Graph Problems}
\titlerunning{Optimal In-place Algorithms for Basic Graph Problems} 

\author[1]{Sankardeep Chakraborty}
\author[2]{Kunihiko Sadakane}
\author[3]{Srinivasa Rao Satti}
\affil[1]{RIKEN Center for Advanced Intelligence Project, Japan\\
  \texttt{sankar.chakraborty@riken.jp}}
 \affil[2]{The University of Tokyo, Japan\\
  \texttt{sada@mist.i.u-tokyo.ac.jp}}
  \affil[3]{Seoul National University, South Korea\\
  \texttt{ssrao@cse.snu.ac.kr}}
\authorrunning{Chakraborty, Sadakane and Satti} 

\Copyright{Chakraborty, Sadakane and Satti}

\subjclass{Dummy classification -- please refer to \url{http://www.acm.org/about/class/ccs98-html}}

\EventEditors{John Q. Open and Joan R. Acces}
\EventNoEds{2}
\EventLongTitle{42nd Conference on Very Important Topics (CVIT 2016)}
\EventShortTitle{CVIT 2016}
\EventAcronym{CVIT}
\EventYear{2016}
\EventDate{December 24--27, 2016}
\EventLocation{Little Whinging, United Kingdom}
\EventLogo{}
\SeriesVolume{42}
\ArticleNo{23}

\begin{document}

\maketitle

\begin{abstract}
We present linear time {\it in-place} algorithms for several basic and fundamental graph problems including the well-known graph search methods (like depth-first search, breadth-first search, maximum cardinality search),  connectivity problems (like biconnectivity, $2$-edge connectivity), decomposition problem (like chain decomposition) among various others, improving the running time (by polynomial multiplicative factor) of the recent results of Chakraborty et al. [ESA, 2018] who designed $O(n^3 \lg n)$ time in-place algorithms for a strict subset of the above mentioned problems. The running times of all our algorithms are essentially optimal as they run in linear time. One of the main ideas behind obtaining these algorithms is the detection and careful exploitation of sortedness present in the input representation for any graph without loss of generality. This observation alone is powerful enough to design some basic linear time in-place algorithms, but more non-trivial graph problems require extra techniques which, we believe, may find other applications while designing in-place algorithms for different graph problems in future.   
\end{abstract}


\section{Introduction}
Inspired by the speedy growth of humongous data set (``big data phenomenon"), {\it space efficient algorithms} are becoming increasingly more crucial than ever before. The dire need of such algorithms is also propelled by the pervasive usage of small specialized handheld devices and embedded systems which come equipped with tiny memory. To design such algorithms a vast array of computational models has already been proposed in the literature. In what follows, we briefly mention a few of them in the order they are historically developed.
\newline
In the {\it read-only memory} model (henceforth {\sf ROM}) where the input is read-only, output is write only, and a limited sized random access read/write work space is available, researchers have designed space efficient algorithms for selection and sorting~\cite{ChanMR13,ElmasryJKS14,Frederickson87,MunroP80,PagterR98}, problems in computational geometry~\cite{AsanoBBKMRS14,AsanoMRW11,BarbaKLSS15,ChanC07,DarwishE14}, and graphs~\cite{AsanoIKKOOSTU14,BanerjeeC0S18,CRS17,ElmasryHK15}  among various others.  In the {\it in-place} model, it is assumed that the input elements are given in an array, and the algorithm may use the input array as working space,
hence the algorithm is allowed to modify the array during its execution. However, after the execution all the input elements should be present in the array 
and the output maybe put in the same array or sent to an output stream. The extra space usage during the entire execution of the algorithm is limited to $O(\lg n)$ bits only. A prime example of an in-place algorithm is the classic heap-sort. Other than in-place sorting \cite{FranceschiniMP07}, searching~\cite{FranceschiniM06,Munro86} and selection \cite{LaiW88}, many in-place algorithms were designed in areas such as computational geometry~\cite{BronnimannCC04} and string algorithms~\cite{FranceschiniM07}. A very recent addition to this long list is the in-place algorithms for the graph problems~\cite{Chakraborty00S18}. Other than these, researchers have also designed space efficient algorithms in {\it (semi)-streaming} models~\cite{AlonMS99,FeigenbaumKMSZ05,MunroP80} and recently introduced {\it restore}~\cite{ChanMR18} and {\it catalytic-space}~\cite{BuhrmanCKLS14} models.


\subsection{Previous Work on Space Efficient Graph Algorithms}
In this paper we exclusively deal with space efficient algorithms for graph theoretic problems. Recent study of space efficient graph algorithms in the streaming and semi-streaming models focused on fundamental problems like connectivity, minimum spanning tree, matching etc. See~\cite{McGregor14} for a detailed survey concerning these results. However study of such algorithms in {\sf ROM} dates back to almost 40 years. In fact there already exists a very rich history of extremely space efficient graph algorithms in {\sf ROM}, and this is captured by the complexity class {\sf L} which contains decision problems that can be solved by a deterministic Turing machine using only
logarithmic amount of work space. A plethora of important algorithmic graph problems actually admit an {\sf L} algorithm~\cite{DattaLNTW09,ElberfeldJT10,ElberfeldK14,ElberfeldS16}, and the most famous among these is the Reingold's~\cite{Reingold08} method for checking {\it st}-reachability in undirected graphs. In spite of being optimally space efficient, a major drawback of these algorithms is that their running time is most often some polynomial of very high degree, and this fact is not surprising, given the results of~\cite{EdmondsPA99,Tompa82} (which state that for problems like directed {\it st}-reachability, if the number of bits available is $o(n)$, then some of the natural algorithmic approaches would require super-polynomial time). This fact alone hinders the practicality of these algorithms in applications. Motivated by this problem, and inspired by the pervasive practical applications of the fundamental graph algorithms, recently there has been a surge of interest in improving the space complexity of graph algorithms without paying too much penalty in the running time i.e., reducing the working space of the classical graph algorithms to $O(n)$ bits with little or no penalty in running time. Thus the goal is to design space-efficient yet reasonably time-efficient graph algorithms on the {\sf ROM}. Generally most of the standard implementations of classical graph algorithms take linear or near-linear running time and use $O(n \lg n)$ (or sometimes $O(m \lg n)$ for graphs with $n$ vertices and $m$ edges) bits. A recent series of papers~\cite{AsanoIKKOOSTU14,BanerjeeC0S18,CRS17,Chakraborty2018,ElmasryHK15} with this point of view showed such results for a vast array of basic graph problems, namely depth-first search (henceforth {\sf DFS}), breadth-first search (henceforth {\sf BFS}), minimum spanning tree (henceforth {\sf MST}), (strong) connectivity, topological sorting, recognizing chordal graphs, bi-connectivity, {\it st}-numbering, shortest path and many others. 

Even if these results are still both time and space efficient, it seems to still require $\Theta(n)$ bits for most of important graph algorithms, and this is a major concern in places with severe space constraints. In order to break this inherent space bound barrier and still obtain reasonable time efficiency, Chakraborty et al.~\cite{Chakraborty00S18} initiated a systematic study of designing efficient {\it in-place} (i.e., using $O(\lg n)$ bits of extra space other than the input space) algorithms for graph problems by defining a new framework which is a slight relaxation of the {\sf ROM}. Using this framework they were also able to show in-place {\sf DFS, BFS, MST}, reachability algorithms taking time $O(n^3 \lg n)$. Despite being optimal in space usage, observe that these results still leave a polynomial gap in the running time from the optimal value. In this work, we essentially obtain the best of the both worlds by closing this gap. More specifically, we show how one can design optimal in-place algorithms i.e., $O(m+n)$ time and using $O(\lg n)$ bits of extra space, for several of these (and a lot more) basic graph algorithms in this work. Recently Kammer et al.~\cite{Kammer} also considered a similar model where they showed efficient in-place algorithms for {\sf DFS}, unordered-{\sf BFS} (will be defined shortly) only.

\subsection{In-place Model for Graph Algorithms and Input Representations}

Before explaining our in-place algorithms and stating main results, in this section we first describe the input graph representation. Note that, as in the case of the standard in-place model, we need to ensure that the graph (adjacency) structure must remain intact throughout entire execution of the algorithm. Let $G=(V,E)$ be the input graph with $n=|V|$, $m=|E|$, and as usual let $V=\{1,2,\cdots,n\}$ denote the vertex set of $G$. We assume that the input graph is given in the standard adjacency array format, and throughout this paper, we refer to this array as $Z$. More specifically, it is an array having size $(n+m+1)$ ($(n+2m+1)$ resp.) words for directed (undirected resp.) graphs where $Z[1]$ stores the number of vertices in $G$, the next $n$ entries (which we refer to as the {\it offsets} part of $Z$) store $n$ pointers (one per vertex) pointing to the location in $Z$ of the last neighbor for each vertex, and finally the last $m$ ($2m$ for undirected graphs) entries are reserved for the edges of $G$. At this point, we should emphasize a small, yet important, technical detail. 
The $Z$ array can be thought of as a single bit array as follows.
For a directed graph $G$, the array $Z$ is a concatenation of $Z[1]$ of length $\ceil{\lg n}$ bits, $Z[2] \ldots Z[n+1]$ of length $\ceil{\lg m}$ bits each\footnote{Note that it is enough to store the offset values starting from $0$, since we can add $n+1$ to the offset value to find the corresponding location in $Z$; hence the offset values can be stored using $\ceil{\lg m}$ bits.}, and finally $Z[n+2] \ldots Z[n+m+1]$ of length $\ceil{\lg n}$ bits each. For undirected graphs, only the second part changes to size $\ceil{\lg m}+1$ bits (instead of $\ceil{\lg m}$) each. Thus, if we just remember the boundaries, we know exactly how many bits we need to read in order to extract useful information from the relevant parts of $Z$. For the sake of simplicity, we drop the ceiling notations from now on. Moreover, throughout this paper, it should be clear from the context the word size depending on which part of $Z$ we are currently working on. See Figure~\ref{figure:graph1} 
for an example. Note that this representation implicitly captures the degree information for every vertex in $G$. Given this format, we say an algorithm $\mathcal{A}$ is an {\it in-place} algorithm if $\mathcal{A}$ (a) may modify any part of $Z$ during its execution, (b) retains all the initial elements of $Z$ (in any order) when it finishes execution; and (c) uses just $O(\lg n)$ bits of extra space. Our goal is to design such algorithms in this paper for a vast array of fundamental graph problems.


In this paper we assume the standard word {\sf RAM} model of computation.
We count space in terms of number of {\it extra} bits used by the algorithm other than the input, and this quantity is referred as ``extra space'' and ``space'' interchangeably throughout the paper. 

\begin{figure}[h]
\begin{center}
	\includegraphics[scale=0.38]{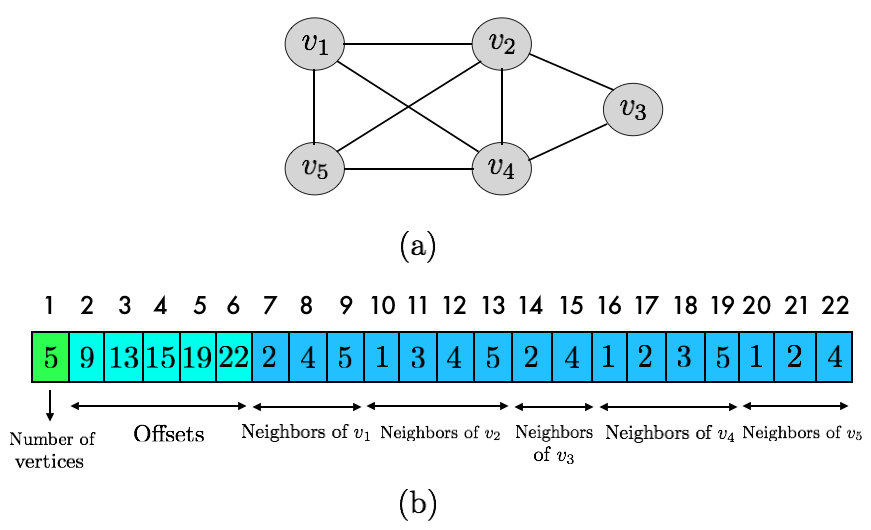}
\end{center}
\caption{(a) An undirected graph $G$ with $5$ vertices and $8$ edges. (b) The standard adjacency array representation of $G$. To avoid cluttering the diagram, we drop the superscript $v$ from the vertex labels while referring to them as neighbors.}
\label{figure:graph1}
\end{figure}

\subsection{Graph Terminology and Notations}
In general we will assume the knowledge of basic graph theoretic terminology as given in~\cite{diestel} and basic graph algorithms as given in~\cite{CLRS}. Still here we collect all the necessary graph theoretic definitions that will be used throughout the paper for quick reference and making the paper self-contained. For {\sf BFS} traversal that we study here, there are two versions studied in the literature. In the {\it ordered {\sf BFS}} (sometimes also known as queue {\sf BFS}~\cite{Chakraborty2018}), vertices are extracted from the queue in the first in first out ({\sf FIFO}) order whereas in the {\it unordered {\sf BFS}}~\cite{BanerjeeC0S18}, vertices can be taken out from the queue in any order as long as no elements are extracted from a higher level of the {\sf BFS} tree before finishing all the vertices from a lower level of the tree. In this paper, by a {\sf BFS/DFS} traversal of the input graph $G$, as in~\cite{AsanoIKKOOSTU14,BanerjeeC0S18,CRS17,Chakraborty2018} we refer to reporting the vertices of $G$ in the {\sf BFS/DFS} ordering, i.e., in the order in which the vertices are visited for the first time. Tarjan et al.~\cite{TarjanY84} defined another method called maximum cardinality search ({\sf MCS}) and used this to give a recognition algorithms for chordal graphs. {\sf MCS} works as follows: assuming that every vertex is unnumbered at the beginning, at each iteration of the execution of {\sf MCS}, an unnumbered vertex that is adjacent to the most number of numbered vertices is chosen (breaking the ties arbitrarily), and is numbered with the next available label. Thus, the output of the {\sf MCS} algorithm is a numbering of the vertices from $1$ to $n$. 

A cut vertex in an undirected graph $G$ is a vertex $v$ that when removed (along with its incident edges) from a graph creates more components than previously in the graph. A (connected) graph with at least three vertices is biconnected if and only if it has no cut vertex. 
Similarly in an undirected graph $G$, a bridge is an edge that when removed (without removing the vertices) from a graph creates more components than previously in the graph. A (connected) graph with at least two vertices is $2$-edge-connected if and only if it has no bridge. 
Given a biconnected graph $G$, and two distinguished vertices $s$ and $t$ in $V$ such that $s \neq t$, $st$-numbering is a numbering of the vertices of the graph so that $s$ gets the smallest number, $t$ gets the largest and every other vertex is adjacent both to a lower-numbered and to a higher-numbered vertex i.e., a numbering $s=v_1,v_2,\cdots,v_n=t$ of the vertices of $G$ is called an $st$-numbering, if for all vertices $v_j, 1<j<n$, there exist $1\leq i<j<k\leq n$ such that $\{v_i, v_j\},\{v_j, v_k\} \in E$. It is well-known that $G$ is biconnected if and only if, for every edge $\{s,t\}\in E$, it has an $st$-numbering. A topological sort of a directed acyclic graph ({\sf DAG}) gives a linear ordering of its vertices such that for every directed edge $(u,v) \in E$ from vertex $u$ to vertex $v$, $u$ comes before $v$ in the ordering. A minimum spanning tree ({\sf MST}) is a subset of the edges of a connected, edge-weighted undirected graph that connects all the vertices together, without any cycles and with the minimum possible total edge weight. That is, it is a spanning tree whose sum of edge weights is as small as possible.


\subsection{Our Main Results and Organization of the Paper}
In Section~\ref{sec:linear-free-bits} we start by designing a linear time in-place procedure to obtain linear bits of additional free space inside the offsets part of the adjacency array. Using this, we can already show an improved set of algorithms for (a strict superset of) problems that Chakraborty et al.~\cite{Chakraborty2018} considered (for example, {\sf DFS}, unordered {\sf BFS} and {\sf MST}), but this algorithms are still not optimal as they are at least polylog multiplicative factor away from linear running time. Towards obtaining optimal linear time in-place algorithms, we first provide an improved linear time in-place routine to obtain almost $n\lg n$ additional free bits of space inside the offsets part, which is what we use crucially along with other additional ideas to show the following main result of this paper in Section~\ref{sec:nlgn-free-bits}. 

\begin{theorem}\label{linear-in-place-results}
Using linear time in the in-place model, one can 
\begin{enumerate}
\item traverse the vertices of any graph in (un)ordered {\sf BFS} and {\sf DFS} manner,
\item recognize bipartite graphs and compute connected components in undirected graphs,
\item report the vertices of a directed acyclic graph ({\sf DAG}) in topologically sorted order,
\item obtain a maximum cardinality search ordering of any graph,
\item output an {\it st}-numbering of given biconnected graph, given two vertices $s$ and $t$, 
\item perform a chain decomposition of any undirected graph, and
\item determine whether any given undirected graph $G$ is biconnected (and/or $2$-edge connected resp.) and if not, we can also compute and report all the cut vertices (bridges resp.) of $G$.
\end{enumerate}
Also, given an undirected edge-weighted (where weights are bounded by some polynomial in $n$) graph $G$, we can find a minimum spanning tree ({\sf MST}) of $G$ in $O(m \lg n)$ time in-place. 
\end{theorem}


\subsection{Techniques}
All the results of our paper stem from the following very simple yet absolutely crucial observation: {\it numbers in sorted order have less entropy than in any arbitrary order}. More specifically, assuming we have $n$ numbers from a universe of size $m$, when these numbers are in any arbitrary order their binary entropy is $n \lg m$ but when they are in sorted order, binary entropy becomes $n \lg m- \Theta(n \lg n)$. This clearly indicates that we can exploit the sorted structure assumption to gain some additional space. Now, note that, without loss of any generality, by construction, the {\it offsets} part of the adjacency array $Z$ for any given graph $G$ is sorted. Thus, we can use the above mentioned idea in the offsets part of $Z$ to gain some free space which is what we use finally to design our optimal in-place graph algorithms. Towards this, we also have to handle several other key technical issues which we describe in respective sections in detail.

\section{Saving Linear Bits and its Applications}\label{sec:linear-free-bits}
As a warm up, in this section we start by showing how we can squeeze in linear sized free bits inside the offsets part of $Z$ while still being able to access any element inside the offsets part in $O(1)$ time, as well as returning to the original configuration of the offsets part of $Z$ before freeing linear bits. Towards this, we first reprove the following lemma, which is essentially same as~\cite[Lemma 5]{Kammer}. See Appendix~\ref{app:lemma2} for a proof.
\begin{lemma}\label{lem:saving-oh-n-bits}
Given a sorted list of $n$ integers from the universe $[0, m-1]$, it can be represented either simply as an array $A[1 . . . n]$ with
the integers in sorted order or as an array of $n$ integers, such that for some fixed constant $c > 1$, the last $cn$ bits of this array are all zero. Moreover, there exists an in-place $O(n)$ time algorithm for switching between both these formats.
\end{lemma}


The above lemma alone is powerful enough to help us design in-place algorithms (albeit with sub-optimal time complexity as we will see shortly) for a variety of fundamental graph algorithms, thanks to the following theorems from~\cite{BanerjeeC0S18,CRS17,Chakraborty2018}.

\begin{theorem}\label{big-oh-rom-results}
In {\sf ROM}, using $cn$ bits of space for some constant $c$, we can 
\begin{enumerate}
\item traverse the vertices of any graph in unordered (ordered resp.) {\sf BFS} manner in $O(m+n)$ ($O(m \lg^2 n)$ resp.) time,
\item recognize bipartite graphs and compute connected components in $O(m+n)$ time, 
\item traverse the vertices of any graph in {\sf DFS} order in $O(m \lg \lg n)$ time,
\item perform a topological sort of a {\sf DAG} in $O(m \lg \lg n)$ time,
\item obtain a maximum cardinality search ordering of any graph in $O(m^2/n+m \lg n)$ time.
\item perform a chain decomposition of any undirected graph in $O(m \lg^2 n \lg \lg n)$ time, 
\item determine whether any given undirected graph $G$ is biconnected (and/or $2$-edge connected resp.) in $O(m \lg n \lg \lg n)$ time, and if not, in the same amount of time and space, we can compute and report all the cut vertices (bridges resp.) of $G$, 
\item output an {\it st}-numbering of given biconnected graph in $O(m \lg^2 n \lg \lg n)$ time, given two distinct vertices $s$ and $t$, and
\item compute a {\sf MST} of a given edge-weighted undirected graph in $O(m \lg n)$ time.
\end{enumerate}
\end{theorem}

Now armed with Theorem~\ref{big-oh-rom-results}, we can easily obtain in-place algorithms for the above mentioned problems almost in a black box manner, albeit with sub-optimal running time (as for all these classical problems linear time algorithms are known~\cite{CLRS,Schmidt13}) except for the unordered {\sf BFS} problem and its applications (as mentioned in item $1$ and $2$) as they already admit linear time. More specifically, given the adjacency array representation (like in the $Z$ array) of the input graph $G$, we first apply Lemma~\ref{lem:saving-oh-n-bits} in the offsets part of $Z$ so that linear bits become free, and this is what we use to store/access the data structures required in Theorem~\ref{big-oh-rom-results}, giving us the same running time bounds for these problems as stated in Theorem~\ref{big-oh-rom-results}. Finally when the execution of these algorithms is finished, we again use Lemma~\ref{lem:saving-oh-n-bits} to restore the original configuration of the offsets part of $Z$, and thereby restoring the $Z$ array completely to its original state (note that the other parts of $Z$ are completely untouched as both Lemma~\ref{lem:saving-oh-n-bits} and Theorem~\ref{big-oh-rom-results} only worked on offsets part of $Z$). Also observe that we use only some constant number of variables (hence $O(\lg n)$ extra bits other than the input) throughout the entire execution of these algorithms. As a consequence, we obtain in-place algorithms for these problems. We summarize our discussion in the theorem below.

\begin{theorem}\label{big-oh-restore-results}
In the in-place model, one can 
\begin{enumerate}
\item traverse the vertices of any graph in unordered (ordered resp.) {\sf BFS} manner in $O(m+n)$ ($O(m \lg^2 n)$ resp.) time,
\item recognize bipartite graphs and compute connected components in $O(m+n)$ time, 
\item traverse the vertices of any graph in {\sf DFS} order in $O(m \lg \lg n)$ time,
\item perform a topological sort of a {\sf DAG} in $O(m \lg \lg n)$ time,
\item obtain a maximum cardinality search ordering of any graph in $O(m^2/n+m \lg n)$ time,
\item perform a chain decomposition of any undirected graph in $O(m \lg^2 n \lg \lg n)$ time,
\item determine whether any given undirected graph $G$ is biconnected (and/or $2$-edge connected resp.) in $O(m \lg n \lg \lg n)$ time, and if not, in the same amount of time and space, we can compute and report all the cut vertices (bridges resp.) of $G$, 
\item output an {\it st}-numbering of given biconnected graph in $O(m \lg^2 n \lg \lg n)$ time, given two distinct vertices $s$ and $t$, and
\item compute a {\sf MST} of a given edge-weighted undirected graph in $O(m \lg n)$ time.
\end{enumerate}
\end{theorem}

Observe that Theorem~\ref{big-oh-restore-results} already improves all the results of Chakraborty et al.~\cite{Chakraborty00S18} as they have shown in-place algorithms for a strict subset of problems (only item $1$, $3$ and $9$ above) considered in Theorem~\ref{big-oh-restore-results} using time $O(n \lg^3 n)$. In the next section, however, we further improve the running time of Theorem~\ref{big-oh-restore-results} by providing improved version of Lemma~\ref{lem:saving-oh-n-bits}.

\section{Saving $n \lg n -2n$ Bits, and Its Applications}\label{sec:nlgn-free-bits}

In what follows, we show how one can improve Lemma~\ref{lem:saving-oh-n-bits} so that almost $n \lg n$ bits become free to be used, and using this we will design optimal in-place algorithms for the above mentioned graph problems. Our main result can be described as follows:

\begin{theorem}\label{thm:saving-nlgn-bits}
Given a sorted list of $n$ integers from the universe $[0, m-1]$, it can be represented either simply as an array $A[1 . . . n]$ with
the integers in sorted order or as an array of $n$ integers, such that the last $n \lg n - 2n$ bits of this array are all zero. Moreover, there exists an in-place $O(n)$ time algorithm for switching between both these formats.
\end{theorem}


\begin{proof}
One can easily obtain the space bound mentioned in the second representation by applying the Elias-Fano encoding~\cite{Elias74,Fano} on the array $A$.
But to implement this encoding in-place, we apply this encoding in two steps as described next.

We split the array $A$ into two subarrays of size $n/2$ each (assume, for simplicity, that $n$ is even) - call them $A_1$ and $A_2$. 
One can replace the most significant $\lg n$ bits of each of the elements in $A_1$ by a bit vector, say $B$, length $n+n/2$, using the Elias-Fano encoding.
To store $B$ (of length $3n/2$), we first replace the most significant $3$ bits of each of the elements in $A_2$ by storing $8$ positions into 
the array $A_2$ (using Lemma~\ref{lem:saving-oh-n-bits}, with $c = 3$). We store the bit vector $B$ inside the most-significant $3$ bits of every element of $A_2$, and 
compact the remaining (least-significant $\lg m - \lg n$) bits of every element in $A_1$ into a consecutive chunk of $(n/2) \lg (m/n)$ bits in $A_1$, so that the
first $(n/2) \lg n$ bits of $A_1$ is free (i.e., filled with all zeros). We now copy the bit vector $B$ into this free space, and restore the $3$ most significant bits
of all the elements of $A_2$. We now replace the most-significant $\lg n$ bits of each element in $A_2$ by a bit vector $C$ of length $3n/2$, and store it inside free space in $A_1$ (here, we assume that $3n \le (n/2) \lg n$), and compact the remaining (least-significant $\lg m - \lg n$) bits into a consecutive
chunk of $(n/2) \lg (m/n)$ bits in $A_2$. Finally, we copy all the lower order bits (of total length $n \lg (m/n)$ bits) into a single chunk, and also merge the two bit vectors of length $3n/2$ each into a single bit vector of length $2n$. Thus the array $A$ is replaced by a total of $n \lg (m/n) + 2n$ bits, giving a free space of $n \lg n - 2n$ bits. These steps can be essentially performed in reverse order to restore the original representation from the second representation. To support the operation of accessing the $i$-th element of $A$ in $O(1)$ time, we can store an additional $o(n)$-bit auxiliary structure that support the 
{\it rank} and {\it select} operations~\cite{Clark96,Munro96} on the $2n$ bit sequence, which can then be used to access the most-significant $\lg n$ bits of any element in $O(1)$ time. The 
remaining $\lg m - \lg n$ bits can be simply read from the array of values stored in the second representation. See Figure~\ref{figure:graph} in Appendix~\ref{missing-picture-in-appendix-for-freespace} for a visual description of the final outcome of application of this theorem.
\end{proof}

In what follows, we show how one can use Theorem~\ref{thm:saving-nlgn-bits} for solving the graph problems mentioned before. Before giving specific details, we would like to sketch the general pattern for designing optimal in-place algorithms for some of these graph problems. Given the adjacency array presentation (as in $Z$) of the input graph $G$, we now first apply Theorem~\ref{thm:saving-nlgn-bits} on the offsets part of $Z$ to make $n \lg n - 2n$ bits free. Now the classical linear time algorithms~\cite{CLRS,EvenT76,Schmidt13,Tarjan72,Tarjan74,TarjanY84} for these problems typically take $cn\lg n + dn$ bits where both the constants $c$ and $d$ are at most $2$. Hence, our idea is to run these algorithms as it is but in some constant number of phases. More specifically, we store only, say $n/3$ vertices, explicitly at any point of time during the execution of these algorithms, and when these vertices are taken care of by the respective algorithms, we refresh the data structures by initiating it with a new set of $n/3$ vertices and proceed again till we exhaust all the vertices, thus, the entire algorithm would finish in three phases ultimately. Now the exact details of refreshing the data structure with a new set of vertices and start the algorithm again where it left off depends on specific problems. This idea would work for most of the algorithms that we discuss in this paper except a few important ones. More specifically, a few of the algorithms for those graph problems are two (or more) pass algorithms, i.e., in the first pass it computes some function which is what used in the second pass to solve the problem finally, for example, chain decomposition, biconnectivity etc. For these kinds of algorithms, it seems hard to make them work using the previously described constant phase algorithmic idea. Thus, we handle them differently by first proving some related lemmata which might be of independent interest, and then use these lemmata to design in-place algorithms for these graph problems. We discuss these after giving proofs for the algorithms which we can handle in constant phases only. In what follows we provide the proofs of linear time in-place algorithms for {\sf DFS} and its applications, especially chain decompsotion, biconnecitivity, $2$-edge connectivity, and also develop/prove the necessary ideas for these algorithms. The missing parts of Theorem~\ref{linear-in-place-results} are proved in Appendix~\ref{missing-proofs-in-appendix-for-linear-inplace} due to lack of space.

The classical implementation of {\sf DFS} (see for example, Cormen et al.~\cite{CLRS}) uses three colors and a stack to traverse the whole graph. More specifically, every vertex $v$ is white initially while it has not been discovered yet, becomes grey when {\sf DFS} discovers $v$ for the first time and pushes on the stack, and is colored black when it is finished i.e., all its neighbors have been explored completely, and it leaves the stack. The algorithm maintains a color array $C$ of length $O(n)$ bits that stores the color of each vertex at any point in the algorithm, along with a stack (which could grow to $O(n \lg n)$ bits) for storing all the grey vertices at any point during the execution. Our idea is to run essentially the same {\sf DFS} algorithm but we limit the stack size so that it contains at most $n/2$ latest grey vertices all the time. More specifically, whenever the stack grows to have more than $n/2$ vertices, we delete the bottom most vertex from the stack so that above invariant is always maintained along with storing the last such vertex to be deleted in order to enforce the invariant. At some point during the execution of the algorithm, when we arrive at a vertex $v$ such that none of $v$’s neighbors are white, then we color the vertex $v$ as black, and we pop it from the stack. If the stack is still non-empty, then the parent of $v$ (in the {\sf DFS} tree) would be at the top of the stack, and we continue the {\sf DFS} from this vertex. On the other hand, if the stack becomes empty after removing $v$, we need to reconstruct it to the state such that it holds the last $n/2$ grey vertices after all the pops done so far. We refer to this phase of the algorithm as reconstruction step. For this, using ideas from~\cite{AsanoIKKOOSTU14}, we basically repeat the same algorithm but with one twist which also enables us now to skip some of the vertices during this reconstruction phase. In detail, we again start with an empty stack, insert the root $s$ first and scan its adjacency list from the first entry to skip all the black vertices and insert into the stack the leftmost grey vertex. Then the repeat the same for this newly inserted vertex into the stack until we reconstruct the last $n/2$ grey vertices. As we have stored the last vertex to be deleted for maintaining the invariant true, we know when to stop this reconstruction procedure. 
It is not hard to see that this procedure correctly reconstructs the latest set of grey vertices in the stack. We continue this process until all the vertices become black. Moreover, this algorithm runs in $O(m+n)$ time as it involves two phases each taking linear time in the worst case, and uses at most $(n \lg n)/2 + n \lg 3$ bits which fits in our budget of free space in the offsets part of the adjacency array. This completes the description of the linear time in-place {\sf DFS} algorithm. 

Before providing the algorithms for other problems, we need a few additional ideas which we will describe next. In the following theorem, we are interested in dynamically maintaining the degree sequence of all vertices that belong to a spanning subgraph of the original graph. More specifically, given a graph $G=(V,E)$, we want to run some algorithm on $G$ for constructing a sparse spanning subgraph $G'=(V, E')$ (which is a spanning subgraph of $G$ i.e., $E' \subseteq E$ and $|E'|=O(V)$) of $G$, and we are interested in dynamically maintaining the degree of all the vertices $v$ in $G'$ i.e., degree of a vertex $v$ in $G'$ is defined as the number of neighbors $u$ such that the edge $(v,u)$ belongs to $G'$. Thus, degree of a vertex $v$ in $G'$ may not be same as degree of $v$ in $G$. Also note that, by the notion of dynamic, we mean that the algorithm starts with an empty graph and gradually add edges to it before finally culminating with a sparse spanning subgraph, thus during the execution of this algorithm degrees of the individual vertices are changing, and it is this dynamically changing degrees that we want to efficiently maintain. We refer to this as the {\it dynamic maintenance of degree sequence} phase. Towards this goal, we prove the following general theorem. 

\begin{theorem}\label{degee-sequence-theorem}
Given a graph $G$ with $n$ vertices and $m$ edges, let $G'$ be a spanning subgraph of $G$ with $m'$ edges, and also let $d' = m'/n$ be the average degree of $G'$. Then, we can construct the dynamically created degree sequence for the vertices of $G'$ in $O(m+n)$ time using $O(n (\lg d' + \lg \lg n))$ bits of construction space. Moreover, the final degree sequence can be stored using $O(n \lg d')$ bits such that degree of any vertex can be returned in $O(1)$ time.
\end{theorem}

\begin{proof}
We divide the vertices into $n/\lg n$ groups of $\lg n$ vertices each. For each group, we allocate a block of $\lg n (\lg d' + \lg \lg n)$ ($\le \lg^2 n$) bits initially (uniformly for all the vertices in the block), to store their degrees. 
We also maintain another parallel bit vector for each block that simply stores the delimiters for each vertex's degree (i.e., a 1 bit to indicate the last bit corresponding to each vertex's degree, and 0 everywhere else). To access the degree of the $i$-th vertex in a block, we first find the positions of the $i-1$-th and the $i$-th 1 bits in the corresponding delimiter sequence, and read the bits between these two positions in the block. 
To perform this efficiently during the construction, we maintain an auxiliary structure that supports {\it select} operation in $O(1)$ time~\cite{Clark96,Munro96}. 
At any point, the representation of each block and delimiter sequence consists of an integral number of words, and these representations are maintained as a collection of ``extendible arrays'' using the structure of~\cite[Lemma~1]{RamanR03}.

At any time, a vertex has some number of bits allocated to store its degree. If the degree of the vertex can be updated in-place, then we first access the position where the degree of the vertex is stored, using the select data structure stored for the corresponding delimiter sequence, and update the degree of the node stored within the block.
Otherwise, we first note that at least $\lg n$ increments have been performed to some vertex within the block (since each vertex has a `slack' of $\lg \lg n$ bits at the beginning of the latest re-construction of the block). Now, we spend $O(\lg n)$ time to re-construct the block (and also the corresponding delimiter sequence with its select structure) so that the degree of each vertex $v$ in the block is stored $\lceil{\lg d_v}\rceil + \lg\lg n$ bits, where $d_v$ is the current degree of $v$. This $\lg n$ construction time can be amortized over the $\lg n$ increments performed on the block before its re-construction, incurring an $O(1)$ amortized cost per increment.
Once we construct the degree sequence for the entire subgraph $G'$, we can scan all the blocks, and compact the degree sequence so that it occupies $O(n \lg d')$ bits. The space usage during the construction is bounded by $O(n (\lg d' + \lg \lg n))$ bits of space. Note that, the above task can be performed while executing the linear time {\sf DFS} algorithm described before, and this completes the proof.
\end{proof}

\begin{corollary}\label{degee-sequence-corollary}
When $G'$ is the DFS tree of $G$, then we can store the dynamically created degree sequence of $G'$, whose size is bounded by $2n$ bits, by running a $O(m+n)$ time DFS procedure while using $O(n \lg \lg n)$ bits of space during construction such that the degree of any vertex in $G'$ can be accessed in $O(1)$ time.
\end{corollary}

For the following discussion, assume that we are working with connected undirected graphs only, and given this, now we are going to describe the {\it setting up parent} phase. More specifically, while performing {\sf DFS}, suppose we visit the vertex $u$ for the first time from the vertex $v$ (hence $v$ becomes the parent of $u$ in the {\sf DFS} tree), at that point we perform one or more swaps in the portion of the adjacency array $Z$ where the neighbors of $u$ are located so that the vertex $v$ becomes the first neighbor of $u$ now. If the initial configuration of $Z$ already satisfies this property in $u$'s neighborhood, we don't need to do anything else. We repeat this procedure for every vertex $v \in V$ so that when {\sf DFS} ends, the first neighbor of every vertex $v$ (except the root vertex) is its parent in the {\sf DFS} tree. Note that we can perform this step of setting up parent in the first location of every neighborhood list of every vertex alongside performing the linear time {\sf DFS} algorithm of Theorem~\ref{linear-in-place-results}. Thus, we obtain the following,

\begin{lemma}\label{setting-up-parent}
There exists a linear time in-place algorithm for performing the setting up parent procedure for every vertex of $G$.
\end{lemma}

Note that, by choosing appropriate parameters, we can actually perform the {\it dynamic maintenance of degree sequence} and the {\it setting up parent} phase together while running the linear time in-place {\sf DFS} algorithm of Theorem~\ref{linear-in-place-results} in any graph $G$. More specifically, suppose we choose to run the linear time in-place {\sf DFS} algorithm of Theorem~\ref{linear-in-place-results} coupled with the setting up parent procedure (to implement Lemma~\ref{setting-up-parent}) by storing $n/2$ vertices (thus taking $n \lg n/2$ bits) in the free space of the offsets part of $Z$, thus, leaving roughly $(n \lg n/2-2n)$ bits of space still free, which can be used to construct and store the degree sequence of all the vertices in the {\sf DFS} tree (to implement corollary~\ref{degee-sequence-corollary}) while running the same linear time in-place {\sf DFS} algorithm of Theorem~\ref{linear-in-place-results}. By degree of a vertex $v$ in the {\sf DFS} tree $T$, we mean the number of children $v$ has in  $T$, and it is this number that gets stored using the algorithm of Corollary~\ref{degee-sequence-corollary}. Hence, at the end of this linear time in-place procedure, we have the following invariant: (a) the first neighbor of every vertex (except the root) is its parent in the {\sf DFS} tree, and (b) the offsets part of $Z$ contains the degree sequence of every vertex $v$ in the {\sf DFS} tree, and this occupies at most $2n$ bits. 

Armed with the above algorithm, we are going to explain next the {\it implicitly representing the search tree} phase. The goal of this phase is to rearrange the neighbors of any vertex $v$ in such a way such that the first neighbor of $v$ becomes its parent in the {\sf DFS} tree (except for the root vertex), this is followed by all of $v$'s children in the {\sf DFS} tree (if any) one by one, finally all the non-child neighbors. Thanks to the setting up parent phase,
we can implement the implicitly representing the search tree phase in linear time overall by doing a reverse search. More specifically, for every non-root vertex $v$, we start by scanning $v$'s list from the second neighbor onward (as first neighbor is its parent), and for each one of them, say $u$, we go to the first location of $u$'s neighbor list to check if $v$ is $u$'s parent if so, we move $u$ in $v$'s list closer to $v$'s parent (i.e., towards the beginning of $v$'s list) by swapping, and repeat this procedure for all the neighbors of $v$'s so that at the end all the children of $v$ are clustered together followed by $v$'s parent. The root vertex can be handled similarly, but we need to start the scanning procedure from the first neighbor itself as it doesn't have any parent. Hence, we spend time proportional to its degree at every vertex, and obtain the following.
\begin{lemma}\label{represent-implicit-tree}
There exists a linear time in-place algorithm for implicitly representing the search tree of $G$.
\end{lemma}

Thus, from now on we can assume that the neighbor list of every vertex is represented in the search tree format implicitly. We choose to call it so as, note that, given in this format, it is very convenient to answer the following queries for any given vertex $v$ in the {\sf DFS} tree $T$: (a) return the parent of $v$ in $T$ in $O(1)$ time, (b) return the number of children $v$ has in $T$ in $O(1)$ time (from the dynamically maintained degree sequence), and finally, (c) enumerate all the children of $v$ one by one optimally in time proportional to its number of children. Not only this, observe that we can still perform the {\sf DFS} traversal of $G$ optimally in linear-time using essentially the same algorithm of Theorem~\ref{linear-in-place-results} given this representation. We can even slightly optimize this {\sf DFS} algorithm by stop scanning the neighbor list of any vertex $v$ as soon as we encounter its last child $u$ in the {\sf DFS} tree (can be derived from the dynamically maintained degree sequence) as neighbors after $u$ will not be of significance in performing the {\sf DFS} traversal of $G$. Hence, we obtain the following.

\begin{lemma}\label{traversal-using-implicit-tree}
There exists a linear time in-place algorithm for performing the {\sf DFS} traversal of a given graph $G$ using the implicit search tree representation of $G$.
\end{lemma}

With the lemma above, we almost complete the description of all the ideas required to design linear time in-place algorithms for the last two items of Theorem~\ref{linear-in-place-results}. In what follows, we describe the necessary graph theoretic background for explaining our algorithms.

Schmidt~\cite{Schmidt2010c} introduced a decomposition of the input undirected graph that partitions the edge set of the graph into cycles and paths, called chains, and used this to design an algorithm to find cut vertices and biconnected components \cite{Schmidt13} and also to test 3-connectivity~\cite{Schmidt2010c} among others. In what follows, we discuss briefly the decomposition algorithm, and state his main result. The algorithm first performs a {\sf DFS} on $G$. Let $r$ be the root of the {\sf DFS} tree $T$ of $G$. {\sf DFS} assigns an index to every vertex $v$, namely, the time vertex $v$ is discovered for the first time during {\sf DFS} -- call it the depth-first-index of $v$ ({\sf DFI(v)}). Imagine that the back edges are directed away from $r$ and the tree edges are directed towards $r$. The algorithm decomposes the graph into a set of paths and cycles called chains as follows. See Figure~\ref{figure:graph-chain} in Appendix~\ref{missing-picture-of-chain-decomposition} for an example. First we mark all the vertices as unvisited. Then we visit every vertex starting at $r$ in the increasing order of {\sf DFI} (i.e., in {\sf DFS} order), and do the following. For every back edge $e$ that originates at $v$, we traverse a directed cycle or a path. This begins with $v$ and the back edge $e$ and proceeds along the tree towards the root and stops at the first visited vertex or the root. During this step, we mark every encountered vertex as visited. This forms the first chain. Then we proceed with the next back edge at $v$, if any, or move towards the next vertex in the increasing {\sf DFI} order and continue the process. Let $D$ be the collection of all such cycles and paths. Notice that the cardinality of this set is exactly the same as the number of back edges in the {\sf DFS} tree as each back edge contributes to a cycle or a path. Also, as initially every vertex is unvisited, the first chain would be a cycle as it would end in the starting vertex. Using this, Schmidt proved the following.

\begin{theorem}[\cite{Schmidt13}]\label{2ec}
Let $D$ be a chain decomposition of a connected undirected graph $G=(V,E)$. Then $G$ is 2-edge-connected if and only if the chains in $D$ partition $E$. Also, $G$
is biconnected if and only if $\delta(G) \geq 2$ (where $\delta(G)$ denotes the minimum degree of $G$) and $D_1$ is the only cycle in the set $D$ where $D_1$ is the first chain in the decomposition. An edge $e$ in $G$ is a bridge if and only if $e$ is not contained in any chain in $D$. A vertex $v$ in $G$ is a cut vertex if and only if $v$ is the first vertex of a cycle in $ D \setminus D_1$.
\end{theorem}

In what follows, we use the linear time in-place {\sf DFS} algorithm of Lemma~\ref{traversal-using-implicit-tree} to perform the tests in Theorem~\ref{2ec}. More specifically, using the linear time {\sf DFS} algorithm of Lemma~\ref{traversal-using-implicit-tree} along with the help of the implicit search tree representation, we can visit every vertex, starting at the root $r$ of the {\sf DFS} tree, in increasing order of {\sf DFI}, and enumerate (or traverse through) all the non-tree (back) edges of the graph as required in Schimdt's algorithm as follows: for each node $v$ in {\sf DFI} order, we can retrieve the number of children, say $t$, $v$ has in the {\sf DFS} tree from the stored degree sequence in $O(1)$ time, then if we directly access the $(t+2)$-th neighbor, say $u$, of $v$, we are guaranteed that the edge $(v,u)$ must be a back edge, thanks to the implicit search tree representation. Moreover all the neighbors starting from the $(t+2)$-th location till the end of $v$'s neighbor list constitute back edges that emanate from $v$. Now we maintain a bit array, {\it visited}, of size $n$ (in the free space of the offsets part), corresponding to the $n$ vertices, initialized to all zeros meaning all the vertices are unvisited at the beginning. We use the {\it visited} array to mark vertices visited during the chain decomposition. When a new back edge is visited for the first time, the algorithm traverses the path starting with the back edge followed by a sequence of tree edges (towards the root) until it encounters a marked vertex, and also marks all the vertices on this path. Note that we can achieve this by repeatedly finding the parent using the implicit search tree representation, i.e., suppose $(v,u)$ is a back edge (which is discovered from $v$'s neighbor list), then we traverse to $u$'s neighbor list to find its parent, say $w$, followed by finding $w$'s parent and so on till we encounter a marked vertex. By checking whether the vertices are marked or not, we can also distinguish whether an edge is encountered for the first time or has already been processed. Note that this procedure constructs the chains on the fly. To check whether an edge is a bridge or not, we first note that only the tree edges can be bridges 
(back edges always form a cycle along with some tree edges). Also, from Theorem~\ref{2ec},
it follows that any (tree) edge that is not covered in the chain decomposition algorithm is a bridge.
Thus, to report these, we maintain a bitvector $M$ of length $n$, corresponding to the $n$ vertices, 
initialized to all zeros. Whenever a tree edge $(u,v)$ is traversed during the chain decomposition
algorithm, if $v$ is the child of $u$, then we mark the child node $v$ in the bit vector $M$. After 
reporting all the chains, we scan the bitvector $M$ to find all unmarked vertices $v$ and output
the edges $(u,v)$, where $u$ is the parent of $v$, as bridges. If there are no bridges found in this process, then $G$ is $2$-edge connected. To check whether a vertex is a cut vertex (using the characterization in 
Theorem~\ref{2ec}), we keep track the starting vertex of the current chain (except for the first chain, which is a cycle), 
that is being traversed, and report that vertex as a cut vertex if the current chain is a cycle. If there are no cut vertices found in this process then $G$ is biconnected. Otherwise, we keep one more array of size
$n$ bits to mark which vertices are cut vertices. This completes the description of linear time in-place algorithms for performing chain decomposition, checking biconnectivity and/or $2$-edge connecitivity, and finding cut vertices and bridges using Schmidt's algorithm.

%



\section{Conclusion}
In this paper, we designed linear time in-place algorithms for a variety of graph problems. As a consequence, many interesting and contrasting observations follow. For example, for {\it directed st-reachability}, the most space efficient polynomial time algorithm~\cite{BarnesBRS98} in {\sf ROM} uses $n/2^{\Theta(\sqrt{\lg n})}$ bits. In sharp contrast, we obtain optimal linear time using logarithmic extra space algorithms for this problem as a simple corollary of both {\sf BFS} and {\sf DFS}. Thus, in terms of workspace this is exponentially better than the best known polynomial time algorithm~\cite{BarnesBRS98} in {\sf ROM}. This provided us with one of the main motivations for designing algorithms in the {\it in-place} model. A somewhat incomparable result obtained by Buhrman et al.~\cite{BuhrmanCKLS14,Koucky16} where they gave an algorithm for {\it directed st-reachability} on catalytic Turing machines in space $O(\lg n)$ with catalytic space $O(n^2 \lg n)$ and time $O(n^9)$. Finally, we conclude by mentioning that we barely scratched the surface of designing in-place graph algorithms with plenty of more to be studied in this model in future. For example, can we design linear time in-place algorithms for testing planarity of a graph? Can we compute the max-flow/min-cut in-place? Can we compute shortest paths between any two vertices of a given graph in-place? We leave these problems as our future directions of study. 

\bibliography{dfs}
\newpage
\appendix
\section{Appendix}\label{appendix}
We provide in this section all the missing proofs and diagrams. The corresponding theorems are repeated for the reader's convenience.


\subsection{Configuration of Adjacency Array Before and After Freeing Space}\label{missing-picture-in-appendix-for-freespace}
\begin{figure}[h]
\begin{center}
	\includegraphics[scale=0.165]{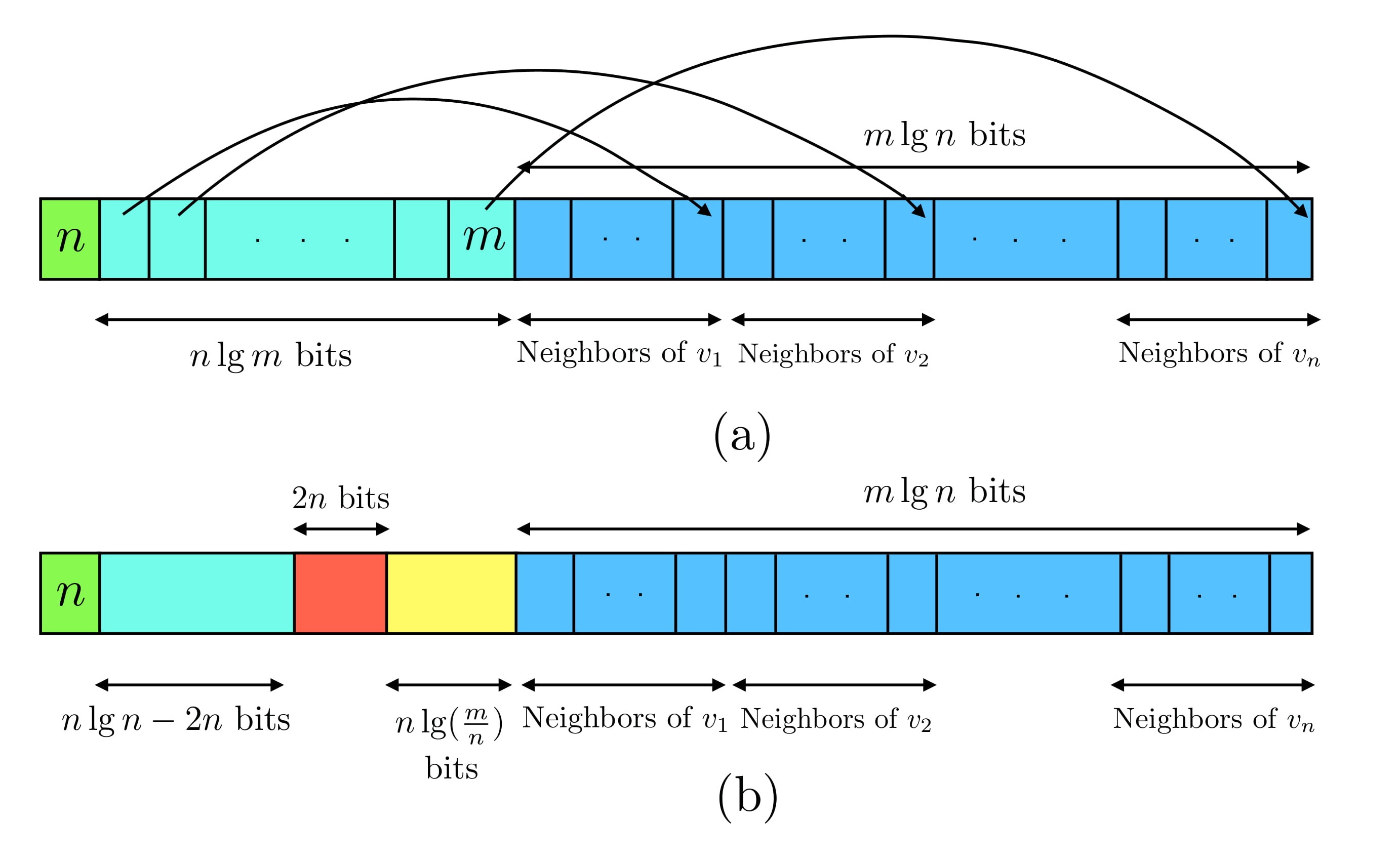}
\end{center}
\caption{(a) General adjacency array structure $Z$ of a given input directed graph. (b) Configuration of $Z$ after freeing $n \lg n-2n$ bits in the offsets part of $Z$.}
\label{figure:graph}
\end{figure}

\subsection{Proof of Lemma~\ref{lem:saving-oh-n-bits}}\label{app:lemma2}
{\bf Lemma~\ref{lem:saving-oh-n-bits} (Restated)}
{\it 
Given a sorted list of $n$ integers from the universe $[0, m-1]$, it can be represented either simply as an array $A[1 . . . n]$ with
the integers in sorted order or as an array of $n$ integers, such that for some fixed constant $c > 1$, the last $cn$ bits of this array are all zero. Moreover, there exists an in-place $O(n)$ time algorithm for switching between both these formats.
}

\begin{proof}
We first note that the most significant $c$ bits 
of all the numbers in any subarray of $A$ can be retrived by storing
$2^c-1$ positions into the subarray, since they form a non-decreasing sequence (of length $n$) over the range $[0, 2^c-1]$. 
Thus if we store these $2^c$ positions in the working space, then we can compact the remaining $\lg m -c$ bits of all the elements of the array
so that they occupy the first $n (\lg m -c)$ bits of the array $A$, freeing the last $cn$ bits. One can decrease the working space usage by performing this in two stages: in the first stage, we can create $n$ bits of empty space by storing one position in the working space (i.e., with $c = 1$); then in the next stage, we can store the remaining $2^c-2$ positions in those linear bits, and create further $(c-1)n$ bits of empty space.
To access the $i$-th element of $A$ in the second representation, we can find the most significant $c$ bits of $A[i]$ by counting the number of positions that are less than $i$ from the $2^c-1$ positions stored. 
The least significant $\lg m -c$ bits can be simply read from the array stored in the second representation. To restore the first representation from the second, one can obtain the most significant $c$ bits of each element from the $2^c$ positions stored, and read the remaining bits from the array.
\end{proof}

\subsection{An Illustration of Chain Decomposition Algorithm}\label{missing-picture-of-chain-decomposition}
\begin{figure}[h]
\begin{center}
	\includegraphics[scale=.6, keepaspectratio=true]{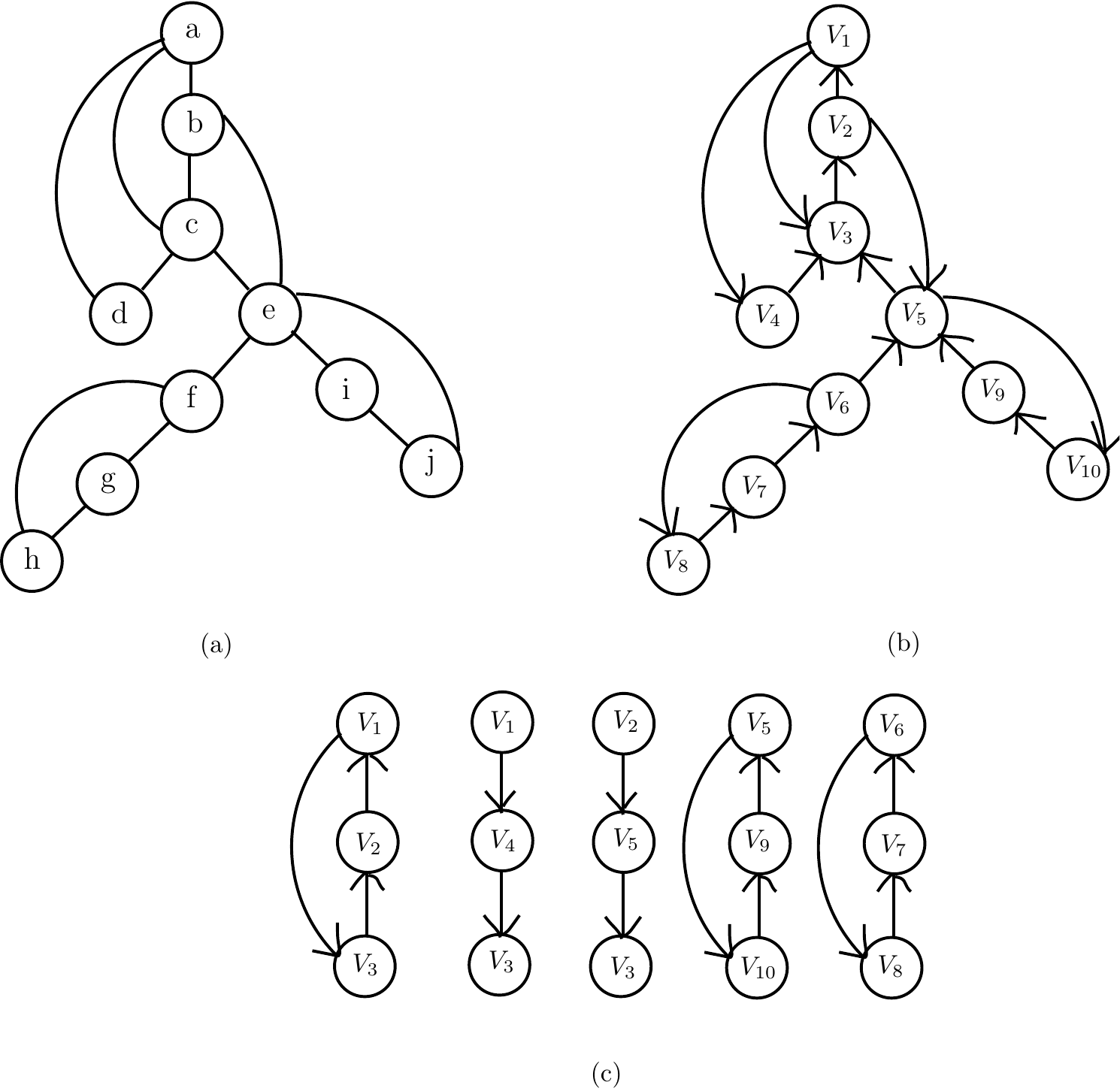}
\end{center}
\caption{Illustration of Chain Decomposition. (a) An input graph $G$. (b) A {\sf DFS} traversal of $G$ and the resulting edge-orientation along with {\sf DFIs}. (c) A chain decomposition $D$ of $G$. The chains $D_2$ and $D_3$ are paths and rest of them are cycles. The edge $(V_5,V_6)$ is bridge as it is not contained in any chain. $V_5$ and $V_6$ are cut vertices.}
\label{figure:graph-chain}
\end{figure}

\subsection{Missing Proofs of Theorem~\ref{linear-in-place-results}} \label{missing-proofs-in-appendix-for-linear-inplace}

{\bf Theorem~\ref{linear-in-place-results} (Restated partly)}
{\it Using linear time in the in-place model, one can 
\begin{enumerate}
\item traverse the vertices of any graph in ordered {\sf BFS} manner ,
\item recognize bipartite graphs and compute connected components in undirected graphs,
\item report the vertices of a directed acyclic graph ({\sf DAG}) in topologically sorted order,
\item obtain a maximum cardinality search ordering of any graph,
\item output an {\it st}-numbering of given biconnected graph, given two vertices $s$ and $t$, 
\end{enumerate}
Also, given an undirected edge-weighted (where weights are bounded by some polynomial in $n$) graph $G$, we can find a minimum spanning tree ({\sf MST}) of $G$ in $O(m \lg n)$ time in-place. }

\begin{proof}
We start with the ordered {\sf BFS} problem for which the classical linear time algorithm starts by coloring all vertices white i.e., unvisited. Then it grows the search starting at the root vertex, say $s$, making it grey (i.e., visited but not completely explored) and adding it to a queue. Finally the algorithm repeatedly removes the first element of the queue, and adds all its white neighbors at the end of the queue (coloring them grey), coloring the element black (i.e., completely explored) after removing it from the queue. As the queue can
store up to $O(n)$ elements, the space for the queue can be $O(n \lg n)$ bits. Towards the efficient in-place implementation of this algorithm, we also maintain (in a slightly different format) all the colors in an array, say $C$, along with maintaining a queue. First observe that elements in the queue are only from two consecutive levels of the {\sf BFS} tree, and the traditional algorithm colors all these vertices grey initially. We refine this further and color them separately as $grey_1$ and $grey_2$. As a consequence, now the algorithm repeatedly removes the first element of the queue (say having color $grey_1$), and adds all its white neighbors at the end of the queue (coloring them $grey_2$), coloring the element black after removing it from the queue, and vice versa. Note that, if the combined size of the two consecutive levels of the {\sf BFS} tree (i.e., $grey_1$ and $grey_2$ vertices) remains bounded by $n/2$ throughout the entire {\sf BFS} algorithm, then we get immediately linear time in-place {\sf BFS} algorithm as we can store these vertices (along with the color array) in the free space of the offsets part of the adjacency array. If not, then one of the levels, say $grey_1$ level without loss of generality and denote it by level $L$, has more than $n/2$ vertices, then all the other levels combined must have less than $n/2$ vertices. In particular, the level $(L-1)$ and $(L+1)$ has less than $n/2$ vertices. As we cannot store the level $L$ explicitly in the free space at once, our idea is to generate the level $L$ from the level $(L-1)$ (which we can store explicitly) {\it on the fly}. More specifically, assuming we have the level $(L-1)$ stored explicitly, we repeatedly expand (till we exhaust) each vertices from this level and generate the vertices of level $L$ but instead of storing them in the queue, we just change the color of these vertices in the color array. Now in order to generate the vertices of level $(L+1)$, we need to expand the vertices of level $L$ in the correct order but we did not store them in the array, thus, we again repeat the expansion of vertices from level $(L-1)$ one by one to generate the vertices of level $L$ in order and at that time only, we process them to generate their neighbors which belong to level $(L+1)$ along with changing their colors in $C$ array, and storing them in a queue (as we can afford to store these) for future exploration. Note that this phenomenon happens at most once during the entire execution of the {\sf BFS} algorithm. Now, it is easy to see that this procedure correctly traverses the input graph $G$ in the ordered {\sf BFS} manner using $O(m+n)$ time in-place. Using this linear time in-place {\sf BFS} algorithm, it is straightforward to obtain the recognition algorithm for bipartite graphs, and compute the connected components in undirected graphs as all of these are simple applications of {\sf BFS}. Hence we also obtain linear time in-place algorithms for these problems.

One of the standard algorithms for computing topological sort~\cite{CLRS} works by simply reporting the vertices of a {\sf DFS} traversal of a given directed acyclic graph in reverse order. We can easily implement this in-place in linear time by running our {\sf DFS} algorithm in two phases. More specifically, in the first phase, we run the {\sf DFS} algorithm completely to generate/store the last $n/2$ vertices in the {\sf DFS} traversal order, and then report them in reverse order. This is followed by running the {\sf DFS} algorithm one more time but stopping just when we obtain the other $n/2$ vertices, then we reverse the order of this vertices and report. This completes the description of generating the vertices in topologically sorted order of an input directed acyclic graph in-place in linear time.

We start by briefly recalling the {\sf MCS} algorithm and its implementation as provided in~\cite{Chakraborty2018,TarjanY84}.
The output of the {\sf MCS} algorithm is a numbering of the vertices from $1$ to $n$. Now assume that every vertex is unnumbered at the beginning of the algorithm. Then, during the execution, at each iteration of the algorithm, an unnumbered vertex that is adjacent to the most number of numbered vertices is chosen (breaking the ties arbitrarily), and is numbered with the next available label. To implement this strategy, the {\sf MCS} algorithm (as described by~\cite{Chakraborty2018} by slightly improving the classical implementation of~\cite{TarjanY84}) maintains an array of sets $set[i]$ for $0 \leq i \leq n-1$ where $set[i]$ stores all unnumbered vertices adjacent to exactly $i$ numbered vertices. So, at the beginning all the vertices belong to $set[0]$. The algorithm also maintains the largest index $j$ such that $set[j]$ is non-empty. To implement an iteration of the {\sf MCS} algorithm, we remove a vertex $v$ from $set[j]$ and number it. For each unnumbered vertex $w$ adjacent to $v$, $w$ is moved from the set containing it, say $set[i]$, to $set[i+1]$. If there is a new entry in $set[j+1]$, we set the largest index to $j+1$, and repeat the same. Otherwise when $set[j]$ becomes empty, we repeatedly decrement $j$ till a non-empty set is found and in this set we repeat the same procedure. In order to delete easily, we implement each set as a doubly-linked list. In addition, for every vertex $v$, we also need to store a pair $(i,j)$ if the vertex $v$ belongs to the list $set[i]$ and $j$ is the pointer to $v$'s location inside $set[i]$ to get linear time overall. This completes the description of the implementation level details of {\sf MCS}. To implement this algorithm in-place using linear time, we essentially run the same algorithm but in four phases. More specifically, we use the above mentioned data structures for storing the top most $n/4$ vertices (in terms of the number of numbered neighbors) at any point during execution along with storing their $(i,j)$ pairs. We also store in a bitvector which vertices have been numbered by the algorithm so far. Then, after we exhaust these vertices, we spend a linear time by scanning the complete adjacency array to determine the next set of top most $n/4$ vertices to initiate these data structures, after which the algorithm starts working where it left off. It can be seen that we can afford to store these information in our free space to essentially ensure linear overall time for the execution of {\sf MCS}. 

For the {\it st}-numbering problem, Chakraborty et al.~\cite{CRS17} modified the classical algorithm of Tarjan~\cite{Tarjan86} and showed that using $O(n)$ bits of space, {\it st}-numbering of any given biconnected graph $G$ and two distinct vertices $s$ and $t$, can be performed in $O(m \lg^2 n \lg \lg n)$ time. This algorithm essentially runs in $O(\lg n)$ phases and in each phase the algorithm performs a {\sf DFS} and some other related work in time $O(m \lg n \lg \lg n)$ while storing $O(n/\lg n)$ vertices in the working space (as they had only $O(n)$ bits available). In order to obtain linear time in-place algorithm for the {\it st}-numbering problem, we can simply run the same algorithm in some constant number, say $k$, of phases along with storing $O(n/k)$ vertices explicitly and in each of those phases, we can run our linear time in-place {\sf DFS} algorithm of Theorem~\ref{linear-in-place-results} to perform the tasks needed in their algorithm~\cite{CRS17}, hence giving us the desired result.

Finally, for the minimum spanning tree ({\sf MST}) problem, we can essentially run the classical Prim's algorithm~\cite{CLRS} using binomial heap in some constant number of phases (like the previous algorithms) to achieve $O(m \lg n)$ running time. We omit the simple details here.  
\end{proof}

\end{document}